\documentclass[11pt]{article}

\usepackage{amsfonts,amssymb,amsmath,latexsym}
\usepackage[vlined,linesnumbered,ruled,resetcount]{algorithm2e}
\usepackage{color,enumitem}
\usepackage{graphicx}
\usepackage{wrapfig,verbatim}

\newcommand{\FF}{\vspace*{\medskipamount}}

\newcommand{\BBB}{\vspace*{-\bigskipamount}}

\newcommand{\Paragraph}[1]{\BBB\paragraph{#1}}
\newcommand{\remove}[1]{}

\newlength{\pagewidth}
\newlength{\captionwidth}
\newcommand{\qed}{\hfill $\square$ \smallbreak}
\newenvironment{proof}{\noindent{\bf Proof:}}{\qed}

\newtheorem{theorem}{Theorem}
\newtheorem{lemma}{Lemma}

\newtheorem{proposition}{Proposition}

\begin{document}

\title{		Stable Scheduling in Transactional Memory\vfill}

\author{	Costas Busch \footnotemark[1]
		\and
		Bogdan S. Chlebus  	\footnotemark[1]   	
		\and
		Dariusz R. Kowalski 	\footnotemark[1]
		\and
		Pavan Poudel \footnotemark[1]}

\footnotetext[1]{School of Computer and Cyber Sciences, Augusta University, Augusta, Georgia, USA.}

\date{}

\maketitle

\vfill

\begin{abstract}
We study computer systems with transactions executed on a set of shared objects.
Transactions arrive continually subjects to constrains that are framed as an adversarial model and impose limits on the average rate of transaction generation and the number of objects that transactions use.
We show that no deterministic distributed scheduler in the queue-free model of transaction autonomy can provide stability for any positive rate of transaction generation.
Let a system consist of $m$ shared objects and an adversary be constrained such that each transaction may access at most $k$ shared objects. 
We prove that no scheduler can be stable if a generation rate is greater than $\max\bigl\{\frac{2}{k+1},\frac{2}{\lfloor \sqrt{2m} \rfloor}\bigr\}$.
We develop a centralized scheduler  that is stable if a transaction generation  rate is at most $\max\bigl\{\frac{1}{4k}, \frac{1}{4\lceil\sqrt{m}\rceil} \bigr\}$. 
We design a distributed scheduler  in the queue-based model of transaction autonomy, in which a transaction is assigned to an individual processor, that guarantees stability if the rate of transaction generation is less than $\max\bigl\{ \frac{1}{6k},\frac{1}{6\lceil\sqrt{m}\rceil}\bigr\}$. 
For each of the schedulers we give upper bounds on the queue size and transaction latency in the range of rates of transaction generation for which the scheduler is stable.

\vfill

~

\noindent
\textbf{Key words:}
Transactional memory, shared object, dynamic transaction generation, adversarial model, stability, latency.
\end{abstract}

\vfill

\thispagestyle{empty}

\setcounter{page}{0}

\newpage

\section{Introduction}

\label{sec:introduction}

We propose to investigate dynamic transactional memory  with arrivals of transactions modeled adversarially.
The goal is to develop a framework to assess the performance of transaction schedulers in quantitative terms.
In particular, we seek schedulers that provide stability of transaction-memory systems and restrain delays in processing transactions.

The adversarial models of generating transactions that we use are inspired by the adversarial queueing theory, which has been applied to study stability of routing algorithms with packets injected continually.
Transmissions of packets in communication networks are constrained by the properties of network, like its topology and capacities of links or channels.
In the case of transactional memory, executing multiple transactions concurrently is constrained by the property that a transaction requires exclusive  access to each object it needs to interact with.

A computer system includes a fixed set of shared objects.
Threads with transactions to execute are spawned continually.
The system is synchronous in that an execution of an algorithm scheduling transactions is structured into rounds. 
It takes one round to execute a transaction successfully.
Multiple transactions can be invoked concurrently, but a transaction requires exclusive access to each object that it needs to interact with in order to be executed successfully.
If multiple transactions accessing the same object are invoked at a round then all of them are aborted.
The arrival of threads with transactions is governed by an adversarial model with parameters bounding the average generation rate and the number of transactions that can be generated at one round.
Processed transactions may be additionally constrained by imposing  an upper bound on the number of objects a transaction needs to access.

The task of a considered computer system is to eventually execute each generated transaction, while striving to limit the number of pending transactions at any round and the time spent waiting by each pending transaction.
Once a transaction is generated, it may need to wait to be invoked.
It is a scheduling algorithm that manages the timings of invocations of pending transactions.
We consider both centralized  and distributed schedulers.

There are two models of generating transactions which specify the autonomy of individual transactions.
In the queue-free case, for each transaction there is a corresponding autonomous thread responsible for its execution.
A distributed scheduler in the queue-free model is executed by the threads that attempt to invoke transactions on shared objects.
In the queue-based model, there is a fixed set of processors, and each thread with a transaction is assigned to a processor.
A distributed scheduler in the queue-based model is executed by the processors that communicate through the shared objects by performing transactions on them.
A centralized scheduler is not affected by constraints on autonomy of each transaction, since all pending transactions are managed en masse.
The schedulers we consider are deterministic, in that they do not resort to randomization.

A scheduler is stable against a given adversary if the number of pending transactions stays bounded at all times, while new transactions are generated by the adversary subject to constrains on the number of new transactions  generated in a contiguous segment of rounds. 
The primary measure of performance of a scheduler is the maximum rate of transaction generation by the adversary such that the scheduler can maintain the system stable for this generation rate.
If a given system is stable, as determined by a set of shared objects and adversarial constraints on transaction generation and a scheduler of transactions, then  upper bounds on the number of pending transactions at a round and a delay of a transaction's execution are used as  secondary performance metrics.

\Paragraph{The contributions.}

We show first that no deterministic distributed scheduler in the queue-free model of transaction autonomy can provide stability for any positive rate of transaction generation.
Let a computer system consist of $m$ shared objects and the adversary be constrained such that each transaction needs to access at most $k$ of the shared objects. 
We show that no scheduler can be stable if a generation rate is greater than $\max\bigl\{\frac{2}{k+1},\frac{2}{\lfloor \sqrt{2m} \rfloor}\bigr\}$.
We develop a centralized scheduler  that is stable if the transaction generate rate is at most $\max\bigl\{\frac{1}{4k}, \frac{1}{4\lceil\sqrt{m}\rceil} \bigr\}$. 
We design a distributed scheduler  in the queue-based model of transaction autonomy, in which a transaction is assigned to an individual processor, that guarantees stability if the rate of transaction generation is less than $  \max\bigl\{ \frac{1}{6k},\frac{1}{6\lceil\sqrt{m}\rceil}\bigr\}$. 
For each of the two schedulers we develop, we give upper bounds on the queue size and transaction latency in the range of rates of transaction generation for which the scheduler is stable.
Table~\ref{table:comparison} gives a summary of the ranges of rates of transaction generation for which deterministic schedulers are stable.

\newcommand{\RB}{\raisebox{3ex}{}}
\newcommand{\LB}{\raisebox{-1.7ex}{}}

\begin{table}[t]
\centering
\begin{tabular}{|c|c|c|}
\hline 
\RB \LB
Scheduler & Lower bound & Upper bound  \\
\hline \hline
\RB \LB
distributed queue-free& stability impossible \hfill (Section~\ref{sec:technical-preliminaries}) &\\
\hline
\RB \LB
centralized & $\rho > \max\bigl\{\frac{2}{k+1},\frac{2}{\lfloor \sqrt{2m} \rfloor}\bigr\}$ \hfill (Section~\ref{sec:lower-bound}) & $\rho\le \max\bigl\{\frac{1}{4k}, \frac{1}{4\lceil\sqrt{m}\rceil} \bigr\}$  \hfill (Section~\ref{sec:centralized-scheduler})\\
\hline
\RB \LB
distributed queue-based & & $\rho <  \max\bigl\{ \frac{1}{6k},\frac{1}{6\lceil\sqrt{m}\rceil}\bigr\}$  (Section~\ref{sec:distributed-scheduler})\\
\hline
\end{tabular}

\caption{\label{table:comparison}
A summary of the ranges of rates of transaction generation for which deterministic schedulers are stable.
The used notations are as follows: $m$ is the number of shared objects, $k$ is the maximum number of shared objects accessed by a transaction, and  $\rho$ is the  rate of transaction generation.
Upper bounds limit transaction generation rates for which stability is achievable.
Lower bounds limit transaction generation rates for which stability is not possible. 
A lower bound for centralized schedulers holds a priori for distributed queue-based schedulers.}
\end{table}

\Paragraph{Related work.}

Scheduling transactions has been studied for both shared memory multi-core and distributed systems. 
Most of the previous work on scheduling transactions considered an offline case where all transactions are known at the outset. 
Some previous work considered online scheduling where a batch of transactions arrives one by one and the performance of an online scheduler is compared to a scheduler processing the batch offline. 
No previous work known to the authors of this paper addressed dynamic transaction arrivals with potentially infinitely many transactions to be scheduled  in a never-ending execution.

Attiya et al.~\cite{AttiyaGM15}, and Sharma and Busch~\cite{SharmaB2014, SharmaB15} considered transaction scheduling in distributed systems with provable performance bounds on communication cost.
Transaction scheduling in a distributed system with the goal of minimizing execution time was first considered by Zhang et al.~\cite{ZhangRP2014}.
Busch et al.~\cite{BuschHPS18} considered  minimizing both the execution time and communication cost simultaneously. 
They showed that it is impossible to simultaneously minimize execution time and communication cost for all the scheduling problem instances in arbitrary graphs even in the offline setting. 
Specifically, Busch et al.~\cite{BuschHPS18} demonstrated a tradeoff between minimizing execution time and  communication cost and provided offline algorithms that separately optimizw  execution time and communication cost. 
Busch et al.~\cite{BuschHPS21} considered transaction scheduling tailored to specific popular topologies and provided offline algorithms that minimize simultaneously execution time and communication cost. 
Busch et al.~\cite{BuschHPS22} studied online algorithms to schedule transactions. 
Distributed directory protocols have been designed by Herlihy and Sun~\cite{HerlihyS2007}, Sharma and Busch~\cite{SharmaB2014}, and Zhang et al.~\cite{ZhangRP2014}, with the goal to optimize communication cost in scheduling transactions. 
Zhang and Ravindran~\cite{ZhangRRG10} provided a distributed dependency-aware model for scheduling transactions in a distributed system that manages dependencies between conflicting and uncommitted transactions so that they can commit safely. 
This model has the inherent tradeoff between concurrency and communication cost. 
Zhang and Ravindran \cite{ZhangRRG10a} provided  cache-coherence protocols for distributed transactional memory based on a distributed queuing protocol.

Adversarial queuing is a methodology to capture stability of processing incoming tasks without any statistical assumptions about task generation.
It provides a framework to develop worst-case bounds on performance of deterministic distributed algorithms in a dynamic setting.
This approach to study routing algorithms in store-and-forward networks was proposed by Borodin et al.~\cite{BorodinKRSW-JACM01}, and continued by Andrews et al.~\cite{AndrewsAFLLK-JACM01}.
Adversarial queuing has been applied to other dynamic tasks in   communication networks.
Bender et al.~\cite{BenderFHKL-SPAA05} considered broadcasting in multiple-access channels with queue-free stations in the framework of adversarial queuing.
Chlebus et al.~\cite{ChlebusKR-TALG12} proposed to investigate deterministic distributed broadcast in multiple access channels performed by stations with queues in the adversarial setting.
This direction was continued by Chlebus et al.~\cite{ChlebusKR-DC09} who studied the maximum throughput in such a setting.
Anantharamu et al.~\cite{AnantharamuCKR-JCSS19} considered packet latency of deterministic broadcast algorithms with injection rates less than~$1$.
Chlebus et al.~\cite{ChlebusHJKK-SPAA19} studied adversarial routing in multiple-access channels subject to energy constraints.
Garncarek et al.~\cite{GarncarekJK19} investigated adversarial stability of memoryless packet scheduling policies in multiple access channels. 
 Garncarek et al.~\cite{GarncarekJK20} studied adversarial communication through  channels with collisions between communicating agents represented as graphs.

\section{Technical Preliminaries}

\label{sec:technical-preliminaries}

A computer system includes a fixed set of $m$ shared objects.
The system executes an algorithm.
An execution of the algorithm is synchronous in that is partitioned into time steps, which we call \emph{rounds}.
The  algorithm spawns threads.
Each thread generates and executes one transaction at a time.
The threads communicate through shared objects.

The \emph{type} of a transaction is the set of objects it may need to access during execution.
To determine the type of a transaction, it suffices to read it to list all the mentioned objects.	 	
The number of objects in a transaction's type is the \emph{weight} of this transaction and the type.
If the types of two transactions share an object, then we say that this creates a \emph{conflict for access} to this object, and that the transactions involved in a conflict for access to an object \emph{collide} at this object.
A set of transactions with the property that no two different transactions in the set collide at some shared object is called \emph{conflict free} or \emph{collision free}.

In a \emph{queue-free} model of transaction autonomy, each transaction is associated with a thread,  which exists only for the purpose to execute this transaction and it disappears after the transaction's execution.
In the \emph{queue-based} model, each transaction is assigned to a group of transactions managed by a processor. 
All the pending transactions at a processor make its \emph{queue}.

\Paragraph{Scheduling transactions.}

Transactions are managed by a \emph{scheduler}.
This is an algorithm that determines for each round which pending transactions are invoked at this round.
A \emph{centralized scheduler} is an algorithm for the queue-free model that knows all the transactions pending at a round, can invoke each pending transaction, and receives feedback from each object about committing to an invoked transaction or aborting it.
In the queue-based model, a \emph{distributed scheduler} is executed by the processors that communicate among themselves through the shared objects.

Scheduling transactions is constrained by whether this is a queue-free or queue-based model.
In a queue-free model, if a pending transaction invoked at a round is not involved in conflict with any object it needs to access, for any of the transactions invoked at this round concurrently, then this transaction is  executed successfully.
It follows that all transactions in a  conflict-free set of transactions can be executed together at one round.
Complementarily, if a pending transaction invoked at a round is involved in conflict for an object it needs to access with some transaction invoked concurrently then both transaction get aborted at this round and stay pending.
The queue-based model is more restricted, in that the queue-free model's  constraints do apply, but additionally, for each processor, at most one transaction in this processor's queue can be performed at a round.
A transaction invoked at a round that executes successfully is no longer pending, while an aborted transaction stays pending in the next round.
In the queue-free model, transactions are managed en-masse.

\Paragraph{Adversaries.}

We consider a setting in which new transactions arrive continuously to the system. 
The process of generation of transactions is represented quantitatively by adversarial models.
We study two types of adversaries corresponding to the queue-free and queue-based models.
In the queue-free model, a transaction generated at a round contributes a unit to the \emph{congestion} at the round at each object the transaction includes in its type.
This is the \emph{queue-free adversary}.
In the queue-based model, a transaction generated at a round at a processor contributes a unit to the \emph{congestion} at the round at each object the transaction includes in its type and also to the processor the transaction is generated at.
This is the \emph{queue-based adversary}.

Quantitative restrictions imposed on adversaries are expressed in terms of bounds on congestion.
A queue-free adversary generates transactions with \emph{generation rate~$\rho$} and \emph{burstiness component~$b$} if, in each contiguous time interval $\tau$ of length $t$ and for each shared object, the amount of congestion created for the object at all the rounds in~$\tau$ together is at most $\rho t + b$.
A queue-based adversary generates transactions with \emph{generation rate~$\rho$} and \emph{burstiness component~$b$} if, in each contiguous time interval $\tau$ of length $t$ and for each shared object and for each processor, the amount of congestion created for the object at all the rounds in~$\tau$ together is at most $\rho t + b$ and the amount of congestion created for the processor at all the rounds in~$\tau$ together is at most $\rho t + b$.
For these adversarial models, we assume that $\rho>0$ is a real number and $b>0$ is an integer.
Each such an adversary is said to be of \emph{type $(\rho,b)$}.
The \emph{burstiness} of an adversary is the maximum number of transactions the adversary can generate in one round. 
By considering a time interval $\tau$ of length~$1$ we obtain that the burstiness of an adversary of type $(\rho,b)$ is $\lfloor \rho + b\rfloor$, which is $b$ for $\rho<1$.

\Paragraph{Performance metrics.}

A scheduler is \emph{stable}, against a given type of an adversary, if the number of pending transactions stays bounded in the course of any execution in which transactions are generated by the adversary of this type.
For an object and a round number~$r$, at most $r$ transactions that contributed to congestion at this object can get executed in the first $r$ rounds.
It follows that no scheduler can be stable if its injection rate is greater than~$1$, so we consider only adversaries of types $(\rho,b)$ in which $0<\rho\le 1$.
A transaction's \emph{delay} is the number of rounds between the generation and execution of this transaction.
The \emph{latency} of a scheduler in an execution is the maximum possible delay of a transaction generated in the execution.


\begin{proposition}

No deterministic distributed scheduler for a system with one shared object can be stable against an adversary of type $(\rho,2)$, for any constant $\rho>0$. 
\end{proposition}

\begin{proof}
Consider an execution in which the adversary generates two transactions that need to access the only existing object at the first round.
This is consistent with the power of the adversary, since its burstiness is~$2$.
Each transaction is controlled by its thread which executes a deterministic code.
We argue by induction of the round number that in each round the states of the threads are identical.
The initialization of each of the threads is the same so they start in the same initial state.
The threads communicate via the shared object.
At a round, assuming that the states of both threads are equal, the threads perform the same action, which is either pausing or invoking the transaction.
Pausing does not generate any feedback from the object, so each thread transitions to the same new state, driven only by the change of the round number.
Invoking a transaction results in receiving feedback indicating conflict for access and the resulting abort, so each thread transitions to the same state determined by this feedback.
Since the states of the threads stay the same at every round, these two transactions will stay pending forever.

If the generation rate $\rho$ is positive, then the adversary can generate new two transactions that have the object in their types again, after sufficiently many rounds of the previous generation have passed.
Such a pattern of repeated generation of pairs of transactions can be repeated indefinitely, resulting in the number of pending transactions growing unbounded. 
\end{proof}

\Paragraph{Coloring graphs.}

We will use properties of colorings of vertices of simple graphs.
An algorithm for coloring vertices of a graph that we call \emph{primary greedy} starts by ordering the vertices in a fixed order.
It will assign colors to vertices in this order.
Positive natural numbers are used as colors. 
A color assigned to a vertex is smallest such that it is different from the colors already assigned to the neighbors.
The primary greedy algorithm uses a number of colors that is at most the maximum vertex degree plus~$1$; such a maximum color is assigned to a vertex~$v$ only if $v$'s degree is maximum, all its neighbors got colors already assigned, and all these colors are all different among themselves.
An algorithm for coloring vertices of a graph that we call \emph{alternative greedy} begins by ordering vertices in a fixed order.
Then maximal independent sets of vertices in induced subgraphs are found in a greedy manner.
To construct a maximal independent set, we process the vertices in the given order one by one.
If no neighbor of a processed vertex has been selected yet  as a member of the independent set, then we add the processed vertex to the independent set, otherwise we skip it.
 After all the vertices have been processed, the vertices in the obtained independent set are given the same color and get removed from the graph to produce a pruned induced subgraph.
 This procedure continues until all the vertices get removed.
 
\begin{proposition}
\label{prop:alternative-greedy}

The alternative greedy coloring algorithm uses a number of colors that is at most the maximum degree plus~$1$.
\end{proposition}
 
\begin{proof}
We may interpret an execution of the alternative greedy coloring with regard to the coloring produced by the primary greedy algorithm.
We show that the ordinal number $x$ of an obtained independent set  consists of all the vertices that get color~$x$ assigned by the primary greedy algorithm.
The proof is by induction on the position of a vertex in a given ordering of vertices.
The base of induction is about the first vertex in the ordering.
This first vertex is colored~$1$ and it also belongs to the first independent set.
To show the inductive step, consider a vertex~$v$ following a block of vertices already colored.
Let $y$ is the smallest available color for vertex~$v$ by the primary greedy coloring.
All the numbers less than $y$ represent independent sets in the order they were built, by the inductive assumption.
Once these independent sets are removed, the vertex~$v$ is placed in the $y$th independent set as its first element.
\end{proof}

\section{A Lower Bound}

\label{sec:lower-bound}

We show that no scheduler can handle dynamic transactions if a generation rate is sufficiently high with respect to the number of shared objects $m$ and an upper bound $k$ on the weight of a transaction.

If $a$ and $b$ are integers where $a\le b$ then let $[a,b]$ denote the set of integers $\{a,a+1,\ldots,b\}$.
We begin with a preliminary fact.


\begin{lemma}
\label{lem:sets}

For an integer $n>0$, there is a family of sets $A_1, A_2,\ldots, A_{n+1}$, each a subset of $[1,\frac{n(n+1)}{2}]$, such that every set $A_i$ has $n$ elements, any two sets $A_i$ and $A_j$, for $i\ne j$, share an element, and each element of $[1,\frac{n(n+1)}{2}]$ belongs to exactly two sets $A_i$ and $A_j$, for $i\ne j$. 
\end{lemma}

\begin{proof}
We start with $A_1=[1,n]$, then set $A_2=\{1\}\cup [n+1,n+(n-1)]$, next set $A_3=\{2,n+1\}\cup [2n,2n-1+(n-2)]$, and next define $A_4=\{3,n+2, 2n\}\cup [3n-2,3n-3+(n-2)]$.
Observe that number~$1$ is shared by $A_1$ and $A_2$, number~$2$ is shared by $A_1$ and $A_3$, number~$3$ is shared by $A_1$ and $A_4$, number~$n+1$ is shared by $A_2$ and $A_3$, number $n+2$ is shared by $A_2$ and $A_4$.
This is the beginning of a general construction that is recursive and proceeds as follows.

Suppose we have defined each among the sets $A_1,\ldots, A_\ell$, for $\ell \le n$.
The set $A_{\ell+1}$ contains $n$ elements, of which the smallest $\ell$ are the following: the first number in $A_1$ that does not belong to any of the sets $A_1,\ldots, A_\ell$ except for~$A_1$, the first number in $A_2$ that does not belong to any of the sets $A_1,\ldots, A_\ell$ except for~$A_2$, and so on, and finally the first number in~$A_\ell$ that does not belong to any of the sets $A_1,\ldots, A_\ell$ except for~$A_\ell$.
The remaining $n-\ell$ elements in~$A_{\ell+1}$ are specified to be the smallest $n-\ell$ positive integers that do not belong to the union of $A_1,\ldots,A_\ell$.
This process defines the sets $A_1,\ldots, A_k,A_{n+1}$.

By the construction of the sets $A_1,\ldots, A_k,A_{n+1}$, any two sets $A_i$ and $A_j$, for $1\le i,j\le n+1$ and $i\ne j$, share exactly one element.
Set $A_1$ contributes $n$ integers, set $A_2$ contributes $n-1$ integers beyond $A_1$, set $A_3$ contributed $n-2$ integers beyond $A_1\cup A_2$, and set $A_n$ contributes one element beyond $A_1\cup A_2\cup\ldots\cup A_{n-1}$. 
The set $A_{n+1}$ does not contribute any new elements.
The union of $A_1,\ldots, A_n$ covers the integers between $1$ and $n+(n-1) + (n-2)+\cdots +1= \frac{n(n+1)}{2}$, which makes the segment of integers $[1,\frac{n(n+1)}{2}]$. 
\end{proof}

Now, we give a lower bound on generation rate to keep scheduling stable.


\begin{theorem}
\label{thm:lower-bound}

A queue-free adversary of type $(\rho,b)$ generating transactions for a system of $m$ objects such that each transaction is of weight at most $k$ can make any scheduling algorithm unstable if injection rate $\rho$ satisfies $\rho>\max\bigl\{\frac{2}{k+1},\frac{2}{\lfloor \sqrt{2m} \rfloor}\bigr\} $. 
\end{theorem}

\begin{proof}
Let the $m$ objects be denoted as $o_1,o_2,\ldots, o_m$.
Suppose first that $\frac{k(k+1)}{2}\le m$.
The transactions to be generated will use only the objects $o_1,o_2,\ldots, o_s$, where $s=\frac{k(k+1)}{2}$.

Let us take the family of sets $A_1, A_2,\ldots, A_{k+1}$ as in Lemma~\ref{lem:sets}, in which $n$ is set to~$k$.
We will use a fixed set of transactions $T_1, T_2,\ldots, T_{k+1}$ defined such that transaction $T_i$ uses object $o_j$ if and only if $j\in A_i$.
In particular, each transaction uses $k$ objects.
The adversary generates these transactions listed in order $L_0,L_1,L_2,\ldots$, where $L_{i-1}$ is the $i$th transaction generated and $L_i$ is a transaction identical to $T_{1+i\bmod (k+1)}$.
The adversary generates new transactions at full power, in the sense that if a transaction can be generated in a given round, subject to the specification of the adversary's type, that a consecutive transaction is generated.
Specifically, consider a round~$r+1$. 
Let $i$ be the highest index of a transaction $L_i$ generated by round $r$. 
Then in round $r+1$ the adversary generates transactions that make a maximal prefix of the sequence $L_{i+1}, L_{i+2},\ldots$ such that the total number of transactions generated by round $r+1$ satisfies the constraints on objects' congestion of type $(\rho,b)$.
The adversary may generate no transaction at a round and it may generate multiple transaction at a round.
For example, the adversary generates exactly the transactions $L_0, \ldots, L_{b-1}$ simultaneously in the first round.

By Lemma~\ref{lem:sets}, at most one transaction can be executed at a round.
The  $k+1$ transactions $T_1, T_2,\ldots, T_{k+1}$ require $k+1$ rounds to have each one executed, one transaction per round. 
Discounting for the burstiness of generation, which is possible due to the burstiness component $b$ in the type~$(\rho, b)$, these transactions can be generated with a frequency of at most one new transaction generated per round if the execution is to stay stable.

The group of transactions $T_1,\ldots,T_{k+1}$  together contribute $2$ to the congestion of each used object, by Lemma~\ref{lem:sets}.
If an execution is stable then the inequality $\rho(k+1)\le 2$ holds.
This gives a bound $\rho\le \frac{2}{k+1}$ on the generation rate of an adversary if the execution is stable.
In case $\rho > \frac{2}{k+1}$, the adversary can generate at least one transaction at every round, and for each round $r$ it can generate two transactions at some round after~$r$.
Such an execution is unstable, because at most one transaction among $T_1,\ldots,T_{k+1}$ can be executed in one round.

Next, consider the case $\frac{k(k+1)}{2}> m$. 
Let $n$ be the greatest positive integer such that $\frac{n(n+1)}{2}\le m$.
We use a similar reasoning as in the case $\frac{k(k+1)}{2}\le m$, with the family of sets $A_1, A_2,\ldots, A_{n+1}$ as in Lemma~\ref{lem:sets}.
In particular, we use a set of transactions $T_1, T_2,\ldots, T_{n+1}$ defined such that transaction~$T_i$ uses an object $o_j$ if and only if $j\in A_i$.
The rules of generating new transactions by the adversary are similar.
We obtain the inequality $\rho\le \frac{2}{n+1}$ by the same argument.
The inequality $\frac{n(n+1)}{2}\le m$ implies $n+1=\lfloor \frac{1}{2}(1+\sqrt{1+8m})\rfloor$, by algebra.
We have the estimates
\[
\frac{2}{n+1} 
= 
\frac{2}{\lfloor \frac{1}{2}(1+\sqrt{1+8m})\rfloor}
\le
\frac{2}{\lfloor \sqrt{2m} \rfloor}
\ .
\]
If $\rho>\frac{2}{\lfloor \sqrt{2m} \rfloor}$ then also $\rho> \frac{2}{n+1}$.
It follows that if the adversary is of a type $(\rho,b)$ such that 
$\rho> \frac{2}{\lfloor \sqrt{2m} \rfloor}$, then this adversary generating transactions at full power can  generate at least one transaction at every round, and for each round $r$ it can generate two transactions at some round after~$r$.
This makes the queue of transactions grow unbounded.
\end{proof}

\section{A Centralized Scheduler}

\label{sec:centralized-scheduler}

We present a scheduling algorithm that  processes all  transactions pending at a round.
The algorithm is centralized in that it is aware of all the pending transactions while selecting the ones to be executed at a round. 
Throughout this Section we assume the queue-free model of autonomy of individual transactions, and the corresponding queue-free adversarial model of transaction generation.

The centralized scheduler identifies a conflict-free set of transactions pending execution that is maximal with respect to inclusion among all pending transactions.
This is accomplished by first ordering all pending transaction on the time of  generation and then processing them greedily one by one in this order.
The word `greedily' in this context means passing over a transaction only when its type includes an object that belongs to the type of a transaction already selected for execution at the current round.


\begin{figure}[t]

\hrule

\FF

\noindent
\texttt{Algorithm} \textsc{Centralized-Scheduler} 

\FF

\hrule

\FF

initialize \texttt{Pending} $\gets$ an empty list

\texttt{for} each round \texttt{do}

\begin{enumerate}[nosep]
\item[]
append all transactions generated in the previous round at the tail of list \texttt{Pending}
\item[]
initialize $\texttt{Execute}\gets$ an empty set\\
\texttt{if} \texttt{Pending} is nonempty \texttt{then} 
\begin{enumerate}[nosep]
\item[]
\texttt{repeat} 
\begin{enumerate}[nosep]
\item[(a)]
\texttt{entry} $\gets$ first unprocessed list item on \texttt{Pending}, starting from head towards tail
\item[(b)]
\texttt{if} \texttt{entry} is conflict-free with all the transactions in \texttt{Execute} \texttt{then} 
\begin{enumerate}[nosep]
\item[]
remove \texttt{entry} from \texttt{Pending} and add it to set \texttt{Execute}
\end{enumerate}
\end{enumerate}
\texttt{until} \texttt{entry} points at the tail of list \texttt{Pending}
\end{enumerate}
\item[]
execute all the transactions in \texttt{Execute}
\end{enumerate}

\FF

\hrule

\caption{\label{alg:centralized}
A pseudocode of the algorithm scheduling all pending transactions en masse.
Transactions pending execution are stored in a list \texttt{Pending} in the order of generation, with the oldest at the head.
The set \texttt{Execute} includes transactions to execute at a round.
It is selected in a greedy manner, prioritizing older transactions over newer and avoiding conflicts for access to shared objects.} 
\end{figure}

The algorithm is called \textsc{Centralized-Scheduler}, its pseudocode  is given in Figure~\ref{alg:centralized}. 
The variable \texttt{Pending} denotes a list of all pending transactions. 
At the beginning of a round, all newly generated transactions are appended to the tail of this list. 
The list is processed in the order from head to tail, which prioritizes transactions on their arrival time, such that those generated earlier  get processed before these generated later. 
The transactions already selected for execution are stored in the set \texttt{Execute}.
If a transaction in \texttt{Pending} is processed, it is checked for conflicts with transactions already placed in the set \texttt{Execute}.
If a processed transaction does not collide with any transaction already in \texttt{Execute} then it is removed from \texttt{Pending} and added to \texttt{Execute}, and otherwise it is passed over. 
After the whole list \texttt{Pending} have been scanned, all the transactions in \texttt{Execute} get executed concurrently.
No invoked transaction is aborted in the resulting execution, because conflicts of transactions are avoided by the process  to add transactions to the set \texttt{Execute}.

\begin{lemma}
\label{lem:execute-pending}

A set of transactions in \texttt{Execute} obtained at the beginning of a round makes a set of transactions that is maximal with respect to inclusion among subsets of transactions in the list \texttt{Pending} that are conflict-free.
\end{lemma}

\begin{proof}
The list \texttt{Pending} is scanned systematically, by line~(a) and the condition controlling the repeat-loop in the pseudocode in Figure~\ref{alg:centralized}.
A processed transaction is added to the set \texttt{Execute} if and only if it is free of conflicts with all the transactions already placed in \texttt{Execute}, by line~(b) in the pseudocode in Figure~\ref{alg:centralized}. 
This shows that the set \texttt{Execute} produced  after completing the repeat loop is conflict free.
 This set is maximal with this property, because a transaction not in \texttt{Execute} that could possibly be added was considered in line~(a) at some point and not added, which means a conflict was detected.
\end{proof}

The list \texttt{Pending} is updated at the very beginning of a round  by adding all newly generated transactions.
At any round in which there are pending transactions, some set of transactions gets invoked successfully, by Lemma~\ref{lem:execute-pending}.
This means each generated transaction is eventually executed.

To assess the efficiency of executing transactions, let us partition an execution of the algorithm \textsc{Centralized-Scheduler} into contiguous \emph{milestone intervals} of rounds, denoted $I_1, I_2, I_3, \dots$, such that the length of each interval equals $4b\cdot \min\{ k, \lceil\sqrt{m}\rceil \}$ rounds.

The transactions pending in a milestone interval $I_{k}$ at any time are categorized into \emph{old} and \emph{new}: the former are those generated prior to the start of $I_{k}$ and the latter are those generated during~$I_{k}$.
We interpret all the old transactions at the beginning of a round of  execution of algorithm \textsc{Centralized-Scheduler} in interval~$I_{j+1}$ as forming a \emph{transaction conflict graph}, or simply \emph{conflict graph} in this Section.
Such a graph has old transactions as vertices and two  old transactions are connected by an edge if and only if they collide.
The first round of the interval ~$I_{j+1}$ has the biggest such a graph among all the rounds in the interval.
Each following round results in pruning the graph of vertices to produce an induced subgraph.
The conflict graph evolves through a sequence of different induced subgraphs in interval~$I_{j+1}$, unless these subgraphs become empty.

Let us assume the adversary generates transactions with a rate $\rho \le \max\bigl\{\frac{1}{4k}, \frac{1}{4\lceil\sqrt{m}\rceil} \bigr\}$ and with a burstiness $b\ge1$.
The total contribution to congestion of an object by transactions generated during a milestone interval can be estimated as follows:
\begin{equation}
\label{eqn:upperbound-centralized}
   \max\bigl\{\frac{1}{4k}, \frac{1}{4\lceil\sqrt{m}\rceil} \bigr\} \cdot 4b\cdot  \min\{ k, \lceil\sqrt{m}\rceil \} +b 
 \le 
 b+b=2b.
\end{equation}
It follows that the sum of weights of transactions generated in a milestone  interval is at most $2bm$.
Since a transaction uses at least one object, the adversary can generate at most $2bm$ transactions in a milestone interval. 

We show the following invariant for all milestone intervals of an execution.

\begin{lemma}[\rm Centralized milestone invariant]
\label{lem:centralized-invariant}

If a generation rate satisfies $\rho\le  \max\bigl\{\frac{1}{4k}, \frac{1}{4\lceil\sqrt{m}\rceil} \bigr\}$, then there are at most $2bm$ pending transactions at the first round of a milestone interval, and all these transactions get executed by the end of the interval.
\end{lemma}

\begin{proof}
The invariant pertains to at most $2bm$ old pending transactions.
We show the centralized milestone invariant by induction on the index of interval.

The base case concerns interval $I_1$.
The invariant holds because there are no transactions generated prior to this interval.
To show the inductive step,  suppose it holds for an interval~$I_j$ and consider interval $I_{j+1}$.
By the inductive assumption, all the old transactions in~$I_{j+1}$ got generated during interval~$I_j$.
It follows that there are at most $2bm$ old transactions at the beginning of~$I_{j+1}$.
We show next that all the old transactions pending in such an interval~$I_{j+1}$ get executed by the end of~$I_{j+1}$.

An execution of algorithm \textsc{Centralized-Scheduler} during interval~$I_{j+1}$ can be interpreted as an execution of the alternative greedy coloring of the conflict graph, as the graph is determined at the first round of~$I_{j+1}$, by Lemma~\ref{lem:execute-pending}.
The consecutive round numbers in interval~$I_{j+1}$ could be interpreted as colors.
By Proposition~\ref{prop:alternative-greedy}, the number of assigned colors is at most the maximum degree plus~$1$.
The old transactions make a prefix of the list of all pending transactions, by how new  transactions get added to list \texttt{Pending}, by the pseudocode in Figure~\ref{alg:centralized}.
This means that new transactions get considered at a round only after all the old transactions had a chance to be scheduled and added to \texttt{Execute}, and so they do not get in the way of old transactions.
The maximum assigned color is the last round in which all old transactions get completed.

Since each transaction uses at most $k$ objects and each object belongs to at most $2b$ old transactions, by the bound~\eqref{eqn:upperbound-centralized}, each old transaction collides with at most $(2b-1)k$ other old transactions. 
It follows that the maximum degree of the conflict graph  is at most $(2b-1)k$.
The alternative greedy coloring assigns at most $(2b-1)k+1$ colors. 
This is also an upper bound on the number of rounds spent to complete executing all old transactions.

To show the inductive step, it suffices to demonstrate that the length of interval $I_{j+1}$ is an upper bound on the number of colors assigned to the vertices of a graph of old transactions by the greedy coloring.
We consider two cases, depending on the relative magnitude of $k$ and $\lceil\sqrt{m}\rceil$.

Suppose first that $k\le\lceil\sqrt{m}\rceil$.
The  length of each milestone interval  is $4bk$, which is strictly greater than the maximum vertex degree. 
This is because the inequality 
\[
4bk>(2b-1)k+1=2bk-k+1
\]
 is equivalent to $2bk>-k+1$, which holds since both $b\ge 1$ and $k\ge 1$.
 
The other case is $k>\lceil\sqrt{m}\rceil$.
Let us call an old transaction \emph{heavy} if its weight is greater than $\lceil \sqrt{m}\rceil$ and \emph{light} if its weight is at most~$\lceil\sqrt{m}\rceil$.
There are at most $2b \lceil\sqrt{m}\rceil$ heavy old transactions, as otherwise the total weight of old transactions would be strictly greater than  $2b \lceil\sqrt{m}\rceil \cdot \sqrt{m}\ge 2bm$,  which is impossible.
Suppose conservatively that if a heavy transaction is scheduled to be executed at a round then this is the only transaction executed at this round.
There are at most $2b \lceil\sqrt{m}\rceil$ such rounds.
The remaining rounds in the interval execute light transactions only.
We interpret these rounds as belonging to an execution of the alternative greedy coloring.
The subgraph induced by light transactions has maximum degree at most $(2b-1)\lceil\sqrt{m}\rceil$, and so  at most $(2b-1)\lceil\sqrt{m}\rceil+1$ colors get assigned, each color representing a round.
To combine the outcomes of these two counts, there are at most $2b \lceil\sqrt{m}\rceil$ rounds needed to execute heavy transactions and at most $(2b-1)\lceil\sqrt{m}\rceil+1$ rounds to execute light transactions, for a total of these many rounds:
\[
2b \lceil\sqrt{m}\rceil + (2b-1)\lceil\sqrt{m}\rceil+1=(4b - 1) \lceil\sqrt{m}\rceil+1
\ .
\]
The length of a milestone interval is $4b\cdot\lceil\sqrt{m}\rceil$ rounds.
By its end all old transactions are completed, because the inequality
\[
(4b - 1) \lceil\sqrt{m}\rceil+1\le 4b\cdot\lceil\sqrt{m}\rceil
\]
is equivalent to $\lceil\sqrt{m}\rceil\ge 1$.
This completes the proof of the inductive step of the centralized milestone invariant.
\end{proof}

We show next that algorithm \textsc{Centralized-Scheduler} is stable and has bounded transaction latency for suitably low transaction generation rates.

\begin{theorem}
\label{thm:centralized-scheduler}

If algorithm \textsc{Centralized-Scheduler} is executed against an adversary of type $(\rho,b)$, such that each generated transaction accesses at most  $k$ objects out of $m$ shared objects available and transaction-generation rate $\rho$ satisfies $\rho\le    \max\bigl\{\frac{1}{4k}, \frac{1}{4\lceil\sqrt{m}\rceil} \bigr\}$, then the number of pending transactions at a round is at most $4bm$ and  transaction latency is at most $8b\cdot \min\{ k, \lceil\sqrt{m}\rceil \}$.
\end{theorem}

\begin{proof}
To estimate the number of transactions pending at a round, let this round belong to a milestone interval~$I_k$.
The number of old transactions at any round of $I_k$ is at most $2mb$, by the centralized milestone invariant formulated as Lemma~\ref{lem:centralized-invariant}.
During the interval~$I_k$, at most $2mb$ new transactions can be generated.
So $2mb+2mb=4mb$ is an upper bound on the number of pending transactions at the round.

To estimate transaction latency, we use the property that a transaction generated in a milestone interval gets executed by the end of the next interval, again by the centralized milestone invariant formulated as Lemma~\ref{lem:centralized-invariant}.
This means that transaction latency is at most twice the length of a milestone  interval, which is $2\cdot4b\cdot \min\{ k, \lceil\sqrt{m}\rceil \}= 8b\cdot \min\{ k, \lceil\sqrt{m}\rceil \}$.
\end{proof}

\section{A Distributed Scheduler}

\label{sec:distributed-scheduler}

We now consider distributed scheduling.
Let a distributed system consist of $n$ processors. 
The processors issue threads that communicate among themselves through some $m$ shared objects.
Every transaction type includes at most $k$ objects.

Each generated transaction is assigned to a specific processor and resides in its local queue while pending execution.
This means we consider the queue-based model of autonomy of individual transactions, and the corresponding queue-based adversarial model of transaction generation.

We employ a specific communication mechanism between a pair of processors.
One of the processors, say $s$, is a sender and the other processor, say~$r$, is a receiver.
The two processors $s$ and $r$ communicate through a designated object~$o$.
Communication occurs at a given round.
All the processors are aware that this particular round is a round of communication from $s$ to~$r$.
Each of the participants $s$ and $r$ may invoke a transaction involving object~$o$ at the round, while at the same time all the remaining processors pause and do not invoke any transactions at this round.

Assume first that both $s$ and~$r$ have pending transactions that access  object~$o$.
At a round of communication, the recipient processor~$r$ invokes a transaction~$t_r$ that uses object~$o$.
If the sender processor~$s$ wants to convey bit~$1$ then $s$ also invokes a transaction~$t_s$ that uses object~$o$.
In this case, both transactions~$t_r$ and~$t_s$ get aborted, so that the processor~$r$ receives the respective feedback from the system and interprets it as receiving~$1$.
If the sender processor~$s$ wants to convey bit~$0$ then $s$ does not invoke any transactions using object $o$ at this round.
In this case, transaction~$t_r$ gets executed successfully, so that $r$ receives the respective feedback from the system and interprets it as receiving~$0$.
This is how one bit can be transmitted successfully from a sender~$s$ to a recipient~$r$.

That was an example of a perfect cooperation between a sender and receiver, but alternative scenarios are possible as well.
Suppose that the sender~$s$ has a pending transaction using object~$o$ and wants to communicate with~$r$ but the recipient~$r$ either does not want to communicate or does not have a pending transaction using object~$o$.
What occurs is that $s$ invokes a suitable transaction $t_s$ which gets executed but $r$ does not receive any information.
Alternatively, suppose that the receiver~$r$ has a pending transaction using object~$o$ and wants to communicate while the sender~$s$ either does not want to communicate or does not have a pending transaction using object~$o$. 
What occurs is that the receiver~$r$ invokes a suitable transaction~$t_r$ which gets executed, which the receiver~$r$ interprets as receiving the bit~$0$. 

That communication mechanism can be extended to transmit the whole type of any  transaction in the following way.
The type identifies a subset of all $m$ objects.
Having a fixed ordering of the objects, the type can be represented as a sequence of $m$ bits, in which $1$ at position $i$ represents that the $i$th object belongs to the type, and $0$ represents that the $i$th object does not belong.
A processor $s$ can transmit a transaction type to recipient $r$ by transmitting $m$ bits representing the type in $m$ successive rounds while using some  designated object~$o$.
We say that by this operation processor~$s$ \emph{sends the transaction type to processor~$r$ via object~$o$}.
This operation works as desired assuming each of the processors has at least $m$ transactions involving object~$o$.
If at least one of these processors either does not have $m$ transactions involving object~$o$ or does not want to participate, then either no bits are  transmitted, or the receiver~$r$ possibly receives a sequence of~$0$s only, which it interprets as no type of transaction successfully transmitted.

Pending transactions at a processor are grouped by their types.
All pending transactions of the same type at a processor make a \emph{block of transactions} of this type.
The \emph{weight of a block} is defined to be the weight of its type.
If there are sufficiently many transactions in a block then the block and the type are said to be \emph{large}.
A boundary number defining sizes to be large is denoted by $L$ and equals $L=(n-1)^2 n^2 m^2$.
If the number of transactions of some type in a queue at a processor is  at least $kL$ but less than $(k+1)L$, for a positive integer~$k$, then we treat these transactions as contributing $k$ large blocks.

An execution of the scheduling algorithm is partitioned into epochs, and each consecutive epoch consists of three phases, labeled Phase~1, Phase~2, and Phase~3.
Each phase is executed the same number of $L=(n-1)^2 n^2 m^2$ rounds.
The algorithm is called \textsc{Distributed-Scheduler} and its pseudocode is given in Figure~\ref{alg:distributed}.

In the beginning of Phase~1, each processor~$v$ that has a large block of transactions of some type, selects one such a block, and this type then is \emph{active} at the processor in the epoch.
A processor that starts Phase~1 with an active type is called \emph{active} in this phase.
Processors store large blocks in the order of generation of their last-added transaction.
Each processor chooses as active a large block that comes first in this order.

The purpose of Phase~1 is to spread the information of active types of all the active processors as widely as possible.
Each active processor uses transactions of its active type for communication.
Such communication involves executing transactions, so a block of transactions of a given type may gradually get smaller.
Once a type of a large block becomes active in the beginning of Phase~1, it stays considered as active for the durations of an epoch, even if the number of transactions in the block becomes less than~$L$.
Phase~1 assigns segments of $(n-1) n^2 m^2$ rounds for each pair of processors $s$ and~$r$ and each object~$o$ to spend with~$s$ acting as sender to~$r$ acting as receiver with communication performed via object~$o$.


\begin{figure}[t]

\hrule

\FF

\noindent
\texttt{Algorithm} \textsc{Distributed-Scheduler} 

\FF

\hrule

\FF

\texttt{Phase 1} : {\em sharing information about large active blocks during $L$ rounds}

\begin{enumerate}[nosep]
\item[] repeat $n-1$ times 
\begin{enumerate}[nosep]
\item[]
\texttt{for} each sender processor $s$ and each recipient processor $r$ and each object $o$ \texttt{do}
\begin{enumerate}[nosep]
\item[]
in a segment of rounds assigned for this selection of $s$, $r$, and $o$:
\begin{enumerate}[nosep]
\item[]
\texttt{if} $v$ is active \texttt{and} this is a round when $s=v$ \texttt{then} 
\begin{itemize}[nosep]
\item[]
act as sender to transmit all relevant information to $r$ via object $o$ 
\end{itemize}

\item[]
\texttt{elseif} $v$ is active \texttt{and} this is a round when $r=v$ \texttt{then} 
\begin{itemize}[nosep]
\item[]
act as recipient to receive all relevant information from $s$ via object $o$
\end{itemize}
\end{enumerate}
\end{enumerate}
\end{enumerate}
\end{enumerate}
\texttt{Phase 2} : {\em executing large blocks of transactions during $L$ rounds}

\begin{enumerate}[nosep]
\item[]
\texttt{if} $v$ is active \texttt{then}
\begin{enumerate}[nosep]
\item[]
select active blocks for execution among those learned in Phase~1 
\end{enumerate}
\texttt{if} $v$ is active \texttt{and} its active block got selected  \texttt{then} 
\begin{enumerate}[nosep]
\item[]
\texttt{for} each among $L$ consecutive rounds \texttt{do}
\begin{enumerate}[nosep]
\item[]
\texttt{if} there is a transaction of the active type in the queue \texttt{then} 
\begin{enumerate}[nosep]
\item[]
invoke such a transaction 
\end{enumerate}
\end{enumerate}
\end{enumerate}
\end{enumerate}
\texttt{Phase 3} : {\em executing remaining transactions by solo processors in $L$ rounds}
\begin{enumerate}[nosep]
\item[] for $L$ consecutive rounds
\begin{enumerate}[nosep]
\item[]
\texttt{if} this is a round among $L/n$ ones assigned to $v$ \texttt{then}

\begin{enumerate}[nosep]
\item[]
\texttt{if} the queue is nonempty \texttt{then} invoke a transaction 
\end{enumerate}
\end{enumerate}
\end{enumerate}
\FF

\hrule

\caption{\label{alg:distributed}
A pseudocode of an epoch for a processor~$v$.
Pending transactions are dispersed among the processors.
Number~$L=(n-1)^2 n^2 m^2$ is the duration of each phase.
In Phase~1, processors $s$ and $r$ use transactions from their active large blocks to implement communication.
A sender processor~$s$ transmits the active type for each processor it knows about. 
In Phase~2, large active blocks are selected for execution in a greedy manner, with blocks ordered by the processors' names.
In Phase~3, each processor gets assigned a unique exclusive contiguous segment of $L/n$ rounds, in which to execute up to $L/n$ transactions from its queue in a first-in first-out manner.
}
\end{figure}

Phase 2 is spent on executing transactions in some active blocks selected such that they do not create conflicts for access to shared objects. 
In the beginning of Phase~2, each processor computes a selection of active large blocks of transactions to execute in Phase~2 among those learned in Phase~1. 
This common selection is computed greedily as follows.
The active types learned in Phase~1 are ordered by the owners' names.
There is a working set of active types selected for execution, which is initialized empty.
The active types are considered one by one.
If a processed active type can be added to the working set without creating a conflict for access to an object, then the type is added to the set, and otherwise it is passed over.
This computation is performed locally by each active processor at the beginning of the first round of Phase~2 and each active processor obtains the same output.
The rounds of Phase~2 are spent on executing the transactions of the active blocks selected for execution.
An active processor whose active large block has been selected executes pending transactions in its selected active block as long as some transactions from the block are still available in the queue or Phase~2 is over, whichever happens earlier.

Phase~3 is spent by each processor executing solo its pending transactions, those that have never been included in large blocks.
Each processor is assigned a unique exclusive contiguous segment of $L/n=(n-1)^2 n m^2$ rounds to execute such transactions.
Transactions are performed in the order of their adding to the queue, with those waiting longest executed before those generated later.

Let $P= \sum_{i=1}^k \binom{m}{i}$ be  the number of possible different transaction types in a system of $m$ shared objects such that a type includes at most $k$ objects. 
We will use the estimate $P \le 2^{\text{H}(\frac{k}{m})\, m}$, for $k\le \frac{m}{2}$, where $\text{H}(x)$ is the binary entropy function $\text{H}(x)=x\lg x + (1-x)\lg (1-x)$  for $0<x<1$.

We partition the rounds of an execution of algorithm \textsc{Distributed-Scheduler} into contiguous \emph{milestone intervals} denoted $I_1, I_2, \ldots$.
Each milestone interval consists of $2 b n P \cdot \min\{ k, \lceil\sqrt{m}\rceil \} $ epochs.
Alternatively, a milestone interval consists of $6 b n L P\cdot \min\{ k, \lceil\sqrt{m}\rceil \} $ rounds, after translating the lengths of epochs into rounds. 

The transactions pending in an interval $I_{k}$ at any time are categorized into \emph{old} and \emph{new}: the former are those generated prior to the start of $I_{k}$ and the latter are those generated during~$I_{k}$.
We also apply this terminology to large blocks of transactions, and categorize them accordingly: a \emph{large old block} of transactions is a large block that consists of only old transactions. 

We define a simple graph called \emph{block conflict graph}, or simply a \emph{conflict graph} in this Section.
Vertices are specified to be pairs $(v,S)$ where $v$ is a processor and $S$ is a large old block at~$v$.
Two vertices are connected by an edge if either they share the same  processor name, as its  first coordinate, or the types of large old blocks, in the second coordinates, share an object. 
Observe that, for each processor~$v$, a subgraph induced by all vertices with the same first coordinate~$v$ is a clique.

The block conflict graph allows to interpret the range of communication during Phase~1 in an epoch.
Namely, Phase~1 results in every processor~$v$ learning the active types in all vertices of the connected component of the block conflict graph to which vertices $(v,S)$ belong.
The duration of a phase $L$ is determined such as to have sufficiently many rounds to accomplish this goal.

\begin{lemma}
\label{lem:communication}

After completion of Phase~1, every processor~$v$  knows the active transaction type  of each processor in the  connected component of the block conflict graph to which vertices with the first coordinate~$v$ belong.
\end{lemma}

\begin{proof} 
In the beginning of an epoch, each processor having a large block selects one large block of transactions  as active.
A block's type can be encoded by a string of bits, of which at least one is a~$1$, as $1$s indicate objects in the type.
The use of an object for communication for each sender-receiver pair is exclusive.
If a receiving processor obtains a string of bits that includes at least one occurrence of~$1$, then this is a legitimate type.
Otherwise, if a receiver does not receive any sequence of bits or it decodes a transmission of a type as a sequence of only $0$s, then this means that no type was communicated.
This demonstrates the correctness of transmitting a type between a pair of processors. 

Next, we show that the length~$L$ of Phase~1 is sufficiently large for the relevant information to successfully propagate to reach every processor.
Any piece of information will need to be transmitted at most $n-1$ times, as there are $n$ processors.
The number of ordered pairs of a sender and a receiver is $n(n-1)$.
A pair of communicating processors may use one among $m$ shared objects at an instance of communication.
The information to propagate pertains to each of the $n$ processors.
A type is encoded by $m$ bits.
The length of Phase~1 is a product of all these numbers, so its value allows for the needed communication to propagate.
\end{proof}

\begin{lemma}
\label{lem:execute-large}

A set of active types obtained in the beginning of Phase~2 by an active  processor is maximal with respect to inclusion among subsets of active types, in the processor's connected component in the block conflict graph, that are conflict free.
\end{lemma}

\begin{proof}
Each active processor knows all the active types of processors in its connected component of the block conflict graph, by Lemma~\ref{lem:communication}.
All the active processors in a connected component work with the same list of active types ordered by the names of active processors.
Each processor selects types of blocks greedily based on the same rules of selecting blocks to add to a working set of blocks to execute.
So each active processor in a connected component produces the same list of active types.

Each active large block~$S$ at a processor~$v$ is represented as a pair $(v,S)$.
This pair belongs to a clique induced of all such pairs that share the first component.
The sub-graph of the block conflict graph induced by the pending large old blocks can be partitioned into such cliques.
Each vertex $(v,S)$ in the conflict graph for an active vertex~$v$ is a neighbor of a vertex $(v,S')$ with an active large old block~$S'$.
It follows that a maximal independent set among the active vertices is also such in the whole conflict graph.
\end{proof}

Let us assume that the adversary generates transactions with a rate $\rho< \max\bigl\{ \frac{1}{6k},\frac{1}{6\lceil\sqrt{m}\rceil}\bigr\}$ and with a burstiness $b\ge1$.
Let us call the number $nLP$ the \emph{bulk} of the system.
The contribution to congestion of any object and of any processor by the transactions generated during a milestone interval, with a generation rate~$\rho$, can be bounded above as follows:
\begin{equation}
\label{eqn:upperbound-distributed}
 \rho\cdot 6 b n L  P \cdot \min\{ k, \lceil\sqrt{m}\rceil \} +b 
 \le bnLP
 \ ,
\end{equation}
assuming the bulk of the system $nLP$ is at least $ \frac{1}{1- \,6\rho \min\{ k, \lceil\sqrt{m}\rceil \} }$.
We refer to this assumption about he bulk of the system by saying that {\em the bulk of the system is sufficiently large for a generation rate~$\rho$}, where the generation rate $\rho$ satisfies $\rho< \max\bigl\{ \frac{1}{6k},\frac{1}{6\lceil\sqrt{m}\rceil}\bigr\}$, for the given parameters $m$ and~$k$.

A large block at a processor contributes $L$ to the congestion of every object of its type, and also a priori to the congestion of the processor.
The number of large blocks contributing to congestion of a processor is this \emph{processor's block congestion}.
The number of large blocks contributing to congestion of an object is this \emph{object's block congestion}.

The contribution to a block congestion of a processor during a milestone interval is at most $bnP=C$, by~\eqref{eqn:upperbound-distributed}, if only the bulk of the system is sufficiently large for a generation rate~$\rho$, where $\rho< \max\bigl\{ \frac{1}{6k},\frac{1}{6\lceil\sqrt{m}\rceil}\bigr\}$.
Similarly, the contribution to a block congestion of an object during a milestone interval is at most $bnP=C$, by~\eqref{eqn:upperbound-distributed}, if only the bulk of the system is sufficiently large for a generation rate~$\rho$, where $\rho< \max\bigl\{ \frac{1}{6k},\frac{1}{6\lceil\sqrt{m}\rceil}\bigr\}$.
The sum of block congestions generated in an interval over all objects is at most $Cm$.
A large block uses at least one object, so the adversary can generate  at most $Cm$ large blocks in a milestone interval. 

Next, we show the following invariant for all milestone intervals of an execution of algorithm \textsc{Distributed-Scheduler}.

\begin{lemma}[\rm Distributed milestone invariant]
\label{lem:distributed-invariant}

For a generation rate $\rho< \max\bigl\{ \frac{1}{6k},\frac{1}{6\lceil\sqrt{m}\rceil}\bigr\}$, and assuming the bulk of the system is sufficiently large with respect to~$\rho$, there are at most $bn^5 m^3 P$ pending transactions at a first round of every milestone interval, and all these transactions  get executed by the end of the interval.
\end{lemma}

\begin{proof}
The invariant concerns old transactions in a milestone interval.
We show the  invariant by induction on the index of a milestone interval.

The base case concerns interval $I_1$.
The invariant holds because there are no transactions generated prior to this interval.
To show the inductive step,  suppose it holds for an interval~$I_j$, and consider interval~$I_{j+1}$.
By the inductive assumption, all the old transactions in~$I_{j+1}$ got generated during interval~$I_j$.

First, let us consider old transactions that do not belong to large old blocks.
Every old transaction that does not belong to large old blocks gets executed during the third phases of the epochs in interval~$I_{j+1}$.
To show this, observe that there are no collisions in the third phases of epoch so it suffices to count the number of such transactions. 
There are fewer than $LP$ such old transactions at a processor, since $L$ transactions of some type make a large block and there are at most $P$ types.
A milestone interval consists of $6 b n L P \cdot \min\{ k, \lceil\sqrt{m}\rceil \}$ rounds, so there are $6 b L P \cdot  \min\{ k, \lceil\sqrt{m}\rceil \} $ rounds assigned to each processor during Phase~3, which is greater than~$LP$.

Next, let us consider large old blocks.
The types of transactions selected for execution in Phase~2 do not collide, so no transaction is ever aborted in Phase~2.
To see this, suppose otherwise, that the two active processors include the same object in their respective active types.
The two processors eventually  communicate in the epoch, because they try each shared object as a communication medium.
This communication is successful when some shared object is used for communication, so the processors learn of their respective active types.

By the inductive assumption, there are at most $Cm$ large old blocks of transactions at the beginning of~$I_{j+1}$.
This means  a total of at most  $b nmP$ large blocks, which makes  $nbmLP$ transactions.
We show next that all large old blocks in interval~$I_{j+1}$ get executed by the end of~$I_{j+1}$.

At the start of the first epoch in the interval~$I_{j+1}$, the block conflict graph is determined by old transactions only, those that make large blocks.
For the sake of the analysis, we consider this very graph during epochs in the milestone interval~$I_{j+1}$, disregarding new transactions generated in the meantime.
The graph evolves through consecutive epochs in the interval.  
Each of the following epochs contributes to pruning the original block conflict graph to produce an induced subgraph of the original graph.
These subgraphs evolve  in interval~$I_{j+1}$ through a sequence of different induced subgraphs, to stabilize when they become empty.

An execution of algorithm \textsc{Distributed-Scheduler} during interval~$I_{j+1}$ can be interpreted as an execution of the alternative greedy coloring of the block conflict graph as it is determined at the first epoch of~$I_{j+1}$, by Lemma~\ref{lem:execute-large}.
The consecutive epoch numbers in interval~$I_{j+1}$ could be interpreted as colors.
By Proposition~\ref{prop:alternative-greedy}, the number of assigned colors is at most the maximum degree of the block conflict graph plus~$1$.
Large blocks of transactions are made active during an epoch in the order of their creation.
The maximum assigned color is the last round in which all large old blocks of  transactions get completed.

The contribution to a processor congestion during a milestone interval is at most $bnP=C$, by~\eqref{eqn:upperbound-distributed}.
Since each transaction uses at most $k$ objects and each object belongs to at most $C$ old blocks, by the bound~\eqref{eqn:upperbound-distributed}, each old block collides with at most $(C-1)k$ other old blocks. 
It follows that the maximum degree of the block conflict graph  is at most $C-1+(C-1)k$.
The alternative greedy coloring assigns at most these many colors:
\begin{equation}
\label{eqn:bound-on-colors}
(C-1)k+C=Ck-k+C\le C(k+1)
\ .
 \end{equation}
This is also an upper bound on the number of epochs spent to complete executing all old blocks.

To show the inductive step, we consider two cases, depending on the relative magnitude of $k$ and $\lceil\sqrt{m}\rceil$.
 It suffices to demonstrate that the number of epochs in the interval $I_{j+1}$ is an upper bound on the number of colors assigned to the vertices of a graph of old blocks by a greedy coloring.

Suppose first that $k\le\lceil\sqrt{m}\rceil$.
The number of epochs of a milestone interval is $2Ck$, which is at least as large as the number of colors of the block conflict graph assigned by the greedy coloring and estimated in~\eqref{eqn:bound-on-colors}. 
This is because the inequality 
\[
2Ck\ge C(k+1)=Ck+C
\]
 is equivalent to $Ck\ge C$, which holds since both $C\ge 1$ and $k\ge 1$.

The other case is $k>\lceil\sqrt{m}\rceil$.
Let us call an old block \emph{heavy} if its weight is greater than $\lceil \sqrt{m}\rceil$ and \emph{light} if its weight is at most~$\lceil\sqrt{m}\rceil$.
There are at most $C \lceil\sqrt{m}\rceil$ heavy old transactions, as otherwise the total weight of old blocks would be strictly greater than  $C \lceil\sqrt{m}\rceil \cdot \sqrt{m}\ge Cm$,  which is impossible.
Suppose, conservatively, that if a heavy block is scheduled to be executed during the second phase of an epoch then this is the only active block executed at this phase of this epoch.
There are at most $C \lceil\sqrt{m}\rceil$ such epochs.
The remaining epochs in the interval execute light blocks during their second phases.
We interpret these epochs/phases as representing an execution of the alternative greedy coloring.
The subgraph induced by vertices of the form $(v,S)$, where $S$ is a light large block, is colored by at most $C(\lceil\sqrt{m}\rceil)$ colors, by~\eqref{eqn:bound-on-colors} and the inequality $\lceil\sqrt{m}\rceil<k$.

To combine all this together, observe that there are at most $C \lceil\sqrt{m}\rceil$ epochs needed to execute heavy transactions in their second phases, and at most $C\lceil\sqrt{m}\rceil$ epochs to execute light transactions in their second phases, for a total of these many epochs: $C \lceil\sqrt{m}\rceil + C\lceil\sqrt{m}\rceil=2C \lceil\sqrt{m}\rceil$.
A milestone interval consists of these many epochs:
\[
2 n P b\cdot\lceil\sqrt{m}\rceil=2C\lceil\sqrt{m}\rceil
\ ,
\]
so by its end all large old blocks get completed during second phases.

This completes the proof of the inductive step of the distributed milestone invariant.
The number of old transactions in a milestone interval is at most $bnmLP$, by~\eqref{eqn:upperbound-distributed}.
This bound means at most  $b n^5 m^3\,P$ old transactions.
\end{proof}

Next, we show that algorithm \textsc{Distributed-Scheduler} is stable and has bounded transaction latency for suitably low transaction generation rates.

\begin{theorem}
\label{thm:distributed-scheduler}

If algorithm \textsc{Distributed-Scheduler} is executed against an adversary of type $(\rho,b)$, such that each generated transaction accesses at most  $k\le \frac{m}{2}$ objects out of $m$ shared objects available, and transaction-generation rate $\rho$ satisfies $\rho < \max\bigl\{ \frac{1}{6k},\frac{1}{6\lceil\sqrt{m}\rceil}\bigr\}$, and the bulk of the system is sufficiently large with respect to $\rho$, then the number of pending transactions at a round is at most $2 b n^5 m^3\,2^{\text{H}(\frac{k}{m})m}$ and  transaction latency is at most $12 b n^5 m^2 \,2^{\text{H}(\frac{k}{m})m} \min\{ k, \lceil\sqrt{m}\rceil \}$.
\end{theorem}

\begin{proof}
To estimate the number of transactions pending at a round, let this round belong to a milestone interval~$I_k$.
The number of old transactions at any round of the interval~$I_k$ is at most $b n^5 m^3 P$, by the distributed milestone invariant formulated as Lemma~\ref{lem:distributed-invariant}.
During the interval~$I_k$, at most $b n^5 m^3P$ new transactions can be generated, again by Lemma~\ref{lem:distributed-invariant}, because they will become old when the next interval begins.
So $2 b n^5 m^3 P\le 2 b n^5 m^3\,2^{\text{H}(\frac{k}{m})m}$  is an upper bound on the number of pending transactions at any round, since $P=\sum_{i=1}^k \binom{m}{i}\le 2^{\text{H}(\frac{k}{m})m}$ for $k\le \frac{m}{2}$.

To estimate the transaction latency, we use the property that a transaction generated in an interval gets executed by the end of the next interval, again by the distributed milestone invariant formulated as Lemma~\ref{lem:distributed-invariant}.
This means that transaction latency is at most twice the length of an interval, which is $2\cdot6 b n L P \min\{ k, \lceil\sqrt{m}\rceil \}$, where $L=(n-1)^2 n^2 m^2$.
We obtain that the latency is at most $12 b n^5 m^2\, 2^{\text{H}(\frac{k}{m})m} \min\{ k, \lceil\sqrt{m}\rceil \} $.
\end{proof}

\bibliographystyle{plainurl}

\bibliography{stable-TM}

\end{document}